\documentclass{article}
\usepackage[utf8]{inputenc}

\usepackage[german,french,english]{babel}
\usepackage[T1]{fontenc}

\usepackage{fullpage}

\usepackage{amsmath}
\usepackage{amssymb}
\usepackage{amsthm}
\usepackage{stmaryrd}

\usepackage{faktor}

\newtheoremstyle{itstyle}
  {1.5\topsep}{1.5\topsep}                                          
  {\itshape}{}{\bfseries}                               
  {.}{.5em}                                  
  {\thmname{#1} \thmnumber{#2}\thmnote{ (#3)}}  
\newtheoremstyle{notitstyle}
  {1.5\topsep}{1.5\topsep}                                          
  {}{}{\bfseries}                               
  {.}{.5em}                                  
  {\thmname{#1} \thmnumber{#2}\thmnote{ (#3)}}  

\theoremstyle{itstyle}
\newtheorem{thm}{Theorem}[section]
\newtheorem{lem}[thm]{Lemma}
\newtheorem{prop}[thm]{Proposition}
\newtheorem{cor}[thm]{Corollary}

\theoremstyle{notitstyle}
\newtheorem{defn}[thm]{Definition}

\newtheorem{rmk}[thm]{Remark}

\usepackage{mathrsfs}

\usepackage{pgfplots}
\pgfplotsset{compat=1.5}

\usepackage[makeroom]{cancel}
\usepackage{graphicx}

\usepackage{tikz}
\usetikzlibrary{decorations.pathmorphing,positioning,arrows,calc,chains,fit,shapes,automata,decorations.pathreplacing,calligraphy}

\usepackage{tikz-cd}
\tikzset{%
    symbol/.style={%
        draw=none,
        every to/.append style={%
            edge node={node [sloped, allow upside down, auto=false]{$#1$}}}
    }
}
\tikzset{
modal/.style={
shorten >=1pt,
shorten <=1pt,
auto,
node distance=1.5cm,
semithick
},
world/.style={circle,draw,minimum size=0.5cm,fill=gray!15},
point/.style={circle,draw,inner sep=0.5mm,fill=black},
reflexive above/.style={->,loop,looseness=7,in=120,out=60},
reflexive below/.style={->,loop,looseness=7,in=240,out=300},
reflexive left/.style={->,loop,looseness=7,in=150,out=210},
reflexive right/.style={->,loop,looseness=7,in=30,out=330}
}

\usepackage{soul} 

\usepackage{sidenotes} 

\usepackage{verbatim}

\usepackage{hyperref}
\hypersetup{
    colorlinks=true,
    linkcolor=blue,
    filecolor=magenta,      
    urlcolor=cyan,
}

\newcommand\twoheaddownarrow{\mathrel{\rotatebox[origin=c]{-90}{$\twoheadrightarrow$}}}

\newcommand\restr[2]{{
  \left.\kern-\nulldelimiterspace 
  #1 
  \right|_{#2} 
  }}

\usepackage{array}

\newcolumntype{M}[1]{>{\centering\arraybackslash}m{#1}}
\newcolumntype{N}{@{}m{0pt}@{}}

\let\epsilon\varepsilon
\newcommand{\nuC}{\nu\text{C}}
\newcommand{\muPL}{\mu\text{PL}}

\usepackage{tablefootnote}
\usepackage{float}
\usepackage{mathtools}

\title{On Transition Constructions for Automata \\ \Large A Categorical Perspective}
\author{Mike Cruchten}
\date{\today}

\begin{document}

\maketitle

\abstract{We investigate the transition monoid construction for deterministic
automata in a categorical setting and establish it as an adjunction.
We pair this adjunction with two other adjunctions to obtain
two endofunctors on deterministic automata, a comonad and
a monad, which are closely related, respectively, to the largest set of equations and the smallest
set of coequations satisfied by an automaton. Furthermore, we give similar transition algebra constructions for lasso and
$\Omega$-automata, and show that they form adjunctions.
We present some initial results on sets of equations and coequations for
lasso automata.}

\section{Introduction}%

The transition monoid construction is a very well-known construction in
automata theory which creates a direct connection between coalgebraic
and algebraic language theory. It can be used, amongst others, to show that
language acceptance by a deterministic finite automaton, is equivalent to
language recognition by a finite monoid. Being able to use both algebraic and
coalgebraic tools to study automata has allowed this theory to accumulate
a wealth of results.

In recent years, the study of automata has also been done through a
category theoretical perspective. Automata are for instance seen as a key
example of a coalgebra. These developments have brought with them a
different perspective on well-known constructions such as for instance
minimisation procedures on automata. Among these categorical approaches,
some work is dedicated to the study of varieties and covarieties of languages.

In this paper we make a connection between the well known transition monoid
construction and recent work on varieties and covarieties (\cite{ballester:2015:dualEquivalenceOfEquationsAndCoequations,salamanca:2015:equationsAndCoequations}). These classes
of languages are defined via sets of equations and coequations which are
satisfied by the transition structure of an automaton. Of particular
interest is the greatest set of equations and the least set of coequations,
which are in a certain sense free and cofree objects.
Rutten et al.\ give constructions of these free and cofree objects and
show that they are functors over certain categories, and that they form
a dual equivalence between congruence quotients and preformations of languages \cite{ballester:2015:dualEquivalenceOfEquationsAndCoequations}.

We show that the transition monoid construction forms a monotone map between
two posetal categories, and that it forms a right adjoint, whose left adjoint
is another well-known construction which takes a monoid homomorphism
and turns it into a deterministic automaton. We combine this Galois connection with some other adjunctions
from the literature to define two endofunctors on the category of deterministic
automata, $\nuC$ and $\muPL$, which to each automaton associates an automaton
corresponding to the greatest set of equations satisfied by the reachable automaton,
and an automaton corresponding to the least preformation of languages which includes
certain languages that can be obtained from the automaton by varying the initial and final states.

After we have set up this configuration, we repeat parts of our construction for lasso
and $\Omega$-automata (these automata provide a coalgebraic way to talk about
$\omega$-languages).
This work tries to make a first step towards studying, in a categorical setting,
sets of equations and coequations for $\Omega$-automata.

Our contributions are summarised as follows:
\begin{itemize}
  \item We show that the transition monoid construction for deterministic automata is a right adjoint
    between two posetal categories, whose left adjoint is the machine
    construction.
  \item We construct two endofunctors $\nuC$ and $\muPL$, a comonad and a monad, on the category of
    deterministic automata. The first associates to each automaton the largest
    set of equations satisfied by the reachable part of the automaton. The second
    associates to each automaton the least preformation of languages which includes
    certain languages that can be obtained from the automaton by varying the initial
    and final states.
  \item We establish a relationship between $\nuC$, $\muPL$ and
    $\mathsf{free}$, $\mathsf{cofree}$ (as defined by Rutten et al. in \cite{ballester:2015:dualEquivalenceOfEquationsAndCoequations}, which correspond to the greatest set of equations and least set of coequations satisfied by an automaton).
  \item We instantiate these ideas to lasso automata by defining sets of
    equations and coequations (as in \cite{ballester:2015:dualEquivalenceOfEquationsAndCoequations}).
    We establish an adjunction through a transition algebra and machine construction,
    and define $\nuC$ and $\muPL$ for lasso automata, where $\nuC$ is again to be seen
    as the greatest set of equations satisfied by the reachable lasso automaton.
  \item Lastly, we show that the transition Wilke algebra construction for $\Omega$-automata from
    \cite{cruchten:2022:omegaAutomata} is functorial and gives rise to
    an adjunction, which is related to the adjunction we obtain for lasso automata.
\end{itemize}

Regarding the monad $\muPL$, for a deterministic automaton with state space
$X$ and accepting states $c$, $\muPL(X,c)$ is the least preformation of
languages which contains the languages $L(x,c)$ for $x\in X$.

\paragraph{Related Work.} Our work is related to work by Chernev et al.\ on
adjunctions for lasso and $\Omega$-automata, which give rise to a monad
similar to $\muPL$.
Other lines of work which
are related are that on equations and coequations \cite{ballester:2015:dualEquivalenceOfEquationsAndCoequations}
and on minimisation as found in \cite{bonchi:2014:algebraCoalgebra,bezhanishvili:2020:minimisation}.

\section{Preliminaries}%

We fix a finite set of letters $\Sigma$ called the \emph{alphabet}. The set
of all finite words $\Sigma^\ast$ is the free monoid over $\Sigma$ and comes
with identity $\epsilon$ (the empty word) and multiplication which is just given
by concatenation of words. We use letters $u,v,w$ for words. The set
$\Sigma^\omega$ corresponds to all infinite words over $\Sigma$, which formally
can be thought of as functions $\omega\to \Sigma$. The infinite words
of the shape $uv^\omega$, where $v^\omega$ ($v\not= \epsilon$) corresponds to concatenating
$v$ infinitely often with itself, are called \emph{ultimately periodic} and
the set of all such words is written $\Sigma^{\text{up}}$. Concatenation of an infinite word by a finite
word on the left is defined by juxtaposition. A pair $(u,v)\in\Sigma^\ast\times\Sigma^+$
is called a \emph{lasso} and we think of it as a representative of the
ultimately periodic word $uv^\omega$. We write the set $\Sigma^\ast\times\Sigma^+$
of all lassos succinctly as  $\Sigma^{\ast +}$.

We use the letters $U,V,W$ for languages of finite words and the letters
$L,K$ for languages of infinite words and also languages of lassos. As is standard, a language
of finite words is regular if it is accepted by a deterministic finite automaton
(or recognised by a finite monoid). Similarly, an $\omega$-language is regular
if it is accepted by a finite nondeterministic B\"{u}chi automaton (or
recognised by a finite Wilke algebra or $\omega$-semigroup).

\begin{defn}[\cite{calbrix:1994:ultimatelyPeriodicWords,cruchten:2022:omegaAutomata}]
  We define $\sim_\gamma$ as the equivalence relation on lassos which is
  given by
  \[
    (u,v)\sim_\gamma (u',v') \iff uv^\omega = u'v'^\omega
  \] for lassos $(u,v),(u',v')\in \Sigma^{\ast +}$. Two lassos which
  are related by $\sim_\gamma$ are called \emph{$\gamma$-equivalent}.
\end{defn}

We now introduce the objects we are studying. We closely follow naming
conventions from \cite{ballester:2015:dualEquivalenceOfEquationsAndCoequations}.
For the first part of the article
we focus on deterministic automata. Define the following endofunctors on the
category Set:
\begin{align*}
  G(X)&=X\times \Sigma & F(X)&=X^\Sigma \\
  G_1(X)&= 1 + G(X) & F_2(X)&=F(X)\times 2.
\end{align*}
Then an automaton $(X,\overline{x},\delta,c)$ give rise to
\begin{enumerate}
  \item a $G$-algebra $(X,\delta^{\dashv}:X\times \Sigma\to X)$,
  \item an $F$-coalgebra $(X,\delta:X\to X^\Sigma)$,
  \item a $G_1$-algebra $(X,\overline{x},\delta^{\dashv})$ and
  \item an $F_2$-coalgebra $(X,\delta,c)$.
\end{enumerate}
We often don't distinguish between $\delta$ and $\delta^{\dashv}$, similarly
we don't distinguish between $c:X\to 2$ and the corresponding subset
$\{x\in X\mid c(x)=1\} $. Moreover,
note that $G\dashv F$, and that $\text{Alg}(G)\cong \text{CoAlg}(F)$, where
$\text{Alg}(G)$ denotes the category of $G$-algebras and $\text{CoAlg}(F)$ the
category of $F$-coalgebras. The transition function $\delta$ can be extended
to a map of type $X\to X^{\Sigma^\ast}$ in the usual way, and we make use of
this fact without additional notation. We sometimes refer to
$(X,\delta)$ as an automaton, $(X,\overline{x},\delta)$ a pointed automaton,
$(X,\delta,c)$ an accepting automaton and $(X,\overline{x},\delta,c)$ a pointed
and accepting automaton (DA or when clear just automaton). Given such a pointed
and accepting automaton, we denote by $L(X,\overline{x},c)$ the language
accepted by it, often dropping $X,\overline{x}$ or $c$ if they are clear from
context. An automaton is called \emph{reachable} if any state can be reached
from the initial state upon input of some finite word.

Towards the second part, we switch to lasso automata. We stick to similar
notation and use the following functors over $\text{Set}^2$ :
\begin{align*}
  G(X_1,X_2)&=(X_1\times \Sigma,X_1\times\Sigma+ X_2\times\Sigma) & F(X_1,X_2)&=(X_1^\Sigma\times X_2^\Sigma,X_2^\Sigma), \\
  G_1(X_1,X_2)&=(1,\emptyset)+G(X_1,X_2) & F_2(X_1,X_2)&=F(X_1,X_2)\times (1,2).
\end{align*}
We usually write a pointed and accepting lasso automaton as a tuple
$(X_1,X_2,\overline{x},\delta_1,\delta_2,\delta_3,c)$ where
$\overline{x}\in X_1$, $\delta_1:X_1\to X_1^\Sigma$, $\delta_2:X_1\to X_2^\Sigma$,
$\delta_3:X_2\to X_2^\Sigma$ and $c:X_2\to 2$. We also define the maps
$\delta_\circ = \delta_2 + \delta_3$ (where we tacitly assume that $X_1\cap X_2=\emptyset$)
and $\delta:X_1\to X_2^{\Sigma^{\ast +}}$
given by $\delta(x,(u,v))=\delta_\circ(\delta_1(x,u),v)$. We use $X$ often
as a shorthand for $(X_1,X_2)$ or for the whole automaton if it is clear from
context. The lasso language associated with a lasso automaton is defined as
$L(X,\overline{x},c)=\{(u,v)\mid \delta(\overline{x},(u,v))\in c\}$. Our conventions
for DAs also apply to lasso automata. A lasso automaton is \emph{saturated}
if it respects $\gamma$-equivalence, i.e. if $(u,v)\sim_\gamma (u',v')$ and
either of the lassos is accepted, then so must the other. A saturated
lasso automaton is also called an $\Omega$-automaton and they act as acceptors
of $\omega$-languages (\cite{ciancia:2019:omegaAutomata}).

Some additional categories which appear in this paper are
\begin{enumerate}
  \item Mon: the category of monoids and monoid homomorphisms,
  \item  $A\twoheaddownarrow \mathcal{D}$ : the coslice category under $A$ whose objects are epimorphisms,
  \item $\mathcal{C}$ : the category of congruences for a specified type of algebra, with inclusion as
    morphisms (see \cite{ballester:2015:dualEquivalenceOfEquationsAndCoequations}),
  \item $\mathcal{P}\mathcal{L}$ : the category of preformation of languages as defined in
    \cite{ballester:2015:dualEquivalenceOfEquationsAndCoequations} (these are complete atomic Boolean
    subalgebras of $2^{\Sigma^\ast}$, which are closed under left and right language derivatives),
  \item $\text{Alg}_r(G)$ : the \emph{reachable} $G$-algebras (this only makes sense
    if we have some notion of reachability for $G$-algebras).
\end{enumerate}

We assume familiarity with standard algebraic, coalgebraic and categorical
notions such as congruence, bisimilarity and adjunctions.

\section{The Transition Monoid and Machine Constructions}

This section investigates the transition monoid and machine construction
which are well-known in formal language theory and create
a direct link between coalgebraic and algebraic language theory. Given
a DFA, one can construct a monoid homomorphism from the free monoid over
$\Sigma$ to a finite monoid (often called the transition monoid of the DFA)
which recognises the language accepted by the DFA. Conversely, given such a
monoid homomorphism, one can construct from it a DFA which accepts the language
recognised by the homomorphism.

We show that these constructions are functorial between suitable categories,
and that the functors form an adjunction (or more precisely a Galois connection).
In our initial treatment, we leave
out any accepting states as they can, without great trouble, be added to the
Galois connection at the end.

Bearing this in mind, we briefly outline the strategy for this section.
In a first step, we identify the categories involved. We then define the
functors on objects and show that they are indeed functorial. Finally, we show
that the functors we defined form a Galois connection.

\subsection{The Functors $T$ and $M$}

It is common to assume that the automata one works with are reachable.
Hence, we focus on reachable pointed automata given by a tuple
$(X,\overline{x},\delta)$ (or simply $(\overline{x},\delta)$ when $X$ is clear
from context). We don't make any restrictions on the size of the
carrier, so $X$ may have infinite cardinality.

Given such a tuple $(X,\overline{x},\delta:X\to X^\Sigma)$, we obtain the
map $\delta^\sharp:\Sigma\to X^X$ given by
$\delta^\sharp(a)(x)=\delta(x)(a)$. As $X^X$ naturally forms a monoid
under function composition and with $\text{id}_X$ as identity,
$\delta^\sharp$ uniquely extends to a monoid homomorphism between
$\Sigma^\ast$ (the free monoid over $\Sigma$) and $X^X$. We
simply denote this extended map also by $\delta^\sharp$, and
it is given by $\delta^\sharp(u)(x)=\delta(x)(u)$. The image
of $\delta^\sharp$ is the \emph{transition monoid} of
the automaton. This describes the object part of a functor, which maps
a reachable pointed automaton to a surjective monoid homomorphism whose domain
is $\Sigma^\ast$ (surjectivity here is the counterpart to reachability for automata).

This suggests that the category $\text{Alg}_r(G_1)$ of reachable $G_1$-algebras
and the category $\Sigma^{\ast}\twoheaddownarrow \text{Mon}$\footnote{coslice under
$\Sigma^\ast$ where we restrict our objects to surjective monoid homomorphisms}
of surjective monoid homomorphisms under $\Sigma^\ast$ would provide two
suitable categories for our endeavour.

Before we concretely define the transition monoid functor, we make some
observations about both categories which we state as lemmata.

\begin{lem}\label{lem:thin}
  The categories $\text{Alg}_r(G_1)$ and $\Sigma^{\ast}\twoheaddownarrow \text{Mon}$
  are both thin (or posetal).
\end{lem}

\begin{proof}
  Let $h_1,h_2:(\overline{x},\delta_X)\to (\overline{y},\delta_Y)$ and
  $x\in X$. By reachability there exists  $u\in \Sigma^\ast$ such that
  $x=\delta_X(\overline{x})(u)$ and hence
  \[
    h_1(x)=h(\delta_X(\overline{x})(u))=\delta_Y(\overline{y})(u)=\ldots =h_2(x).
  \] 
  Similarly, let $h_1,h_2:A\to B$ be $\Sigma^\ast\twoheaddownarrow\text{Mon}$
  morphisms between $f:\Sigma^\ast \twoheadrightarrow A$ and $g:\Sigma^\ast\twoheadrightarrow B$.
  Then for any $x\in A$, there exists some $u\in \Sigma^\ast$ such that
  $x=f(u)$. Hence
  \[
  h_1(x)=h_1(f(u))=g(u)=\ldots =h_2(x).\qedhere
  \] 
\end{proof}

Additionally, it is straight-forward to show for either category that any
morphism between two objects must be surjective (seen as a Set morphism).
Moreover, the category $\Sigma^\ast\twoheaddownarrow\text{Mon}$ is equivalent
to another category we have introduced before.

\begin{lem}
  The categories $\Sigma^\ast\twoheaddownarrow\text{Mon}$ and
  $\mathcal{C}$ are equivalent.
\end{lem}

\begin{proof}
  To see this, note that each surjective monoid homomorphism
  gives a congruence (the kernel), and each congruence defines a surjective monoid
  homomorphism (which is unique up to isomorphism).
\end{proof}

We choose to work with the
category $\mathcal{C}$ instead of $\Sigma^\ast\twoheaddownarrow\text{Mon}$ but
this is only personal preference. It creates a first link to the work in
\cite{ballester:2015:dualEquivalenceOfEquationsAndCoequations}.

With this setup, we introduce the two functors $T$ (for transition monoid)
and $M$ (for machine) which are given below:

\begin{minipage}[c][][c]{\textwidth}
  \begin{minipage}{0.45\textwidth}
    \[
    T : \text{Alg}_r(G_1)\to \mathcal{C}
    \] 
    \begin{itemize}
      \item Objects: $(\overline{x},\delta) \mapsto \ker \delta^\sharp$
      \item Morphisms:
        \[
          h:(\overline{x},\delta_X)\twoheadrightarrow (\overline{y},\delta_Y) \implies \ker \delta_X^\sharp \subseteq \ker\delta_Y^\sharp
        \]
    \end{itemize}
  \end{minipage}
  \begin{minipage}{0.45\textwidth}
    \[
    M : \mathcal{C} \to \text{Alg}_r(G_1)
    \] 
    \begin{itemize}
      \item Objects: $C \mapsto ([\epsilon]_C,\sigma_C([w]_C)(a)=[wa]_C)$
      \item Morphisms:
        \[
          h: C \subseteq D \mapsto Mh([w]_C)=[w]_D
        \]
    \end{itemize}
  \end{minipage}
  \vspace{1em}
\end{minipage}

Showing that these are well-defined is straight-forward. For instance, showing
that the object part of $T$ is well-defined just amounts to showing that
$x \overset{u,v}{\longrightarrow} y \overset{u',v'}{\longrightarrow} z$ 
implies $x \overset{uu',vv'}{\longrightarrow} z$ which is trivially the case.
In order to verify that $T$ is well-defined for morphisms, one should make use
of the fact that both automata are reachable.

At this stage one could already bring acceptance back into the picture, in
which case the categories involved would be that of reachable pointed
and accepting automata, and congruence relations with a subset of
the equivalence classes. The functors naturally extend as, for $T$, each choice of
accepting states $c$ gives the subset of equivalence classes
$\{[u]_{\ker \delta^\sharp}\mid \delta(\overline{x})(u)\in c\}$, and for $M$,
given a subset of equivalence classes, this immediately defines a subset of
accepting states. Additionally, one can at this point check that the two
functors are language preserving; for each reachable pointed and acceptable
automaton  $\mathcal{A} =(\overline{x},\delta,c)$ : $L(\mathcal{A})=L(T(\mathcal{A}))$,
and for each congruence relation $C$ over $\Sigma^\ast$ together with a set $P$
of $C$-equivalence classes: $L(M(C,P))=\bigcup P=L(C,P)$.

\subsection{$T\vdash M$}

The functors both being in place, we show that $T$ is a right adjoint of
$M$ (so they form a Galois connection), the proof of which is not very involved.
In fact, the unit is trivial as $C=TMC$ so $\eta_C=\text{id}_C$. To see this,
note that $(u,v)\in TMC$ if for all $[w]_C: \sigma_C([w]_C)(u)=\sigma_C([w]_C)(v)$.
By the definition of $\sigma_C$, this is equivalent to $[wu]_C=[wv]_C$ and
for $w=\epsilon$ we get $[u]_C=[v]_C$, i.e. $(u,v)\in C$, the converse following
with $C$ being a congruence.

The counit of the Galois connection is given by
\begin{align*}
  \epsilon_X: MTX &\longrightarrow X \\
  [u]_\sim &\longmapsto \epsilon_X([u]_\sim) = \delta(\overline{x},u).
\end{align*}

\begin{lem}
  The counit is well-defined.
\end{lem}

\begin{proof}
Let $u\sim v$, then by definition $\delta^\sharp(u)=\delta^\sharp(v)$
and so $\delta(\overline{x})(u)=\delta(\overline{x})(v)$ so the morphism is
well-defined. Moreover, this is a reachable $G_1$-morphism as it preserves the initial state
($\epsilon_X([\epsilon]_\sim)=\delta(\overline{x})(\epsilon)=\overline{x}$)
and respects the transition function as for $a\in \Sigma$:
\[
  \epsilon_X(\sigma_\sim([u]_\sim)(a))=\delta(\overline{x})(ua)=\delta(\epsilon_X([u]_\sim))(a).\qedhere
\]
\end{proof}

\begin{cor}
  $T$ is a right adjoint of $M$.
\end{cor}

As each section of $\eta$ is the identity,
$\mathcal{C}$ can be seen as a full coreflective subcategory of $\text{Alg}_r(G_1)$.
As before, this Galois connection can be extended to incorporate acceptance.
In that case $\eta_C$ is the identity on the subset of equivalence classes
that comes with a congruence, and $\epsilon_X$ maps the set of accepting states
$P\subseteq \Sigma^\ast/{\sim}$ to the set $\{\delta(x,v)\mid [v]_\sim \in P\}$.
This leaves us with the final diagram for this section:

\begin{center}
\begin{tikzpicture}[modal]
\node[] at (3,0) (1) [label=above:{}] {$\text{Alg}_r(G_1)$};
\node[] at (6,0) (2) [label=above:{}] {$\mathcal{C}\cong \Sigma^\ast\twoheaddownarrow\text{Mon}$};
\node[rotate=90] at (4.5,0) {$\dashv$};
\path[->] (1) edge[bend left] node[]{$T$} (2);
\path[->] (2) edge[bend left] node[]{$M$} (1);

\node[] at (3,2.5) (3) [label=above:{}] {$\text{DA}_r$};
\node[] at (6,2.5) (4) [label=above:{}] {$(\mathcal{C},P)$};
\node[rotate=90] at (4.5,2.5) {$\dashv$};
\path[->] (3) edge[bend left] node[]{$\overline{T}$} (4);
\path[->] (4) edge[bend left] node[]{$\overline{M}$} (3);
\path[->] (3) edge (1);
\path[->] (4) edge (2);

\end{tikzpicture}
\end{center}

\section{$\nuC$ and $\muPL$}

In \cite{ballester:2015:dualEquivalenceOfEquationsAndCoequations}, Rutten et al.\ introduce for each automaton
$(X,\delta)$ the automata $\mathsf{free}(X,\delta)$ and
$\mathsf{cofree}(X,\delta)$, corresponding respectively to the largest
set of equations and the smallest set of coequations satisfied by
$(X,\delta)$. They furthermore show, that $\mathsf{free}$ and $\mathsf{cofree}$ 
are functors on suitable categories and that they are dually equivalent over the
categories $\mathcal{C}$ of congruence quotients and $\mathcal{P} \mathcal{L}$ 
of preformations of languages.

In this section we introduce two functors $\nuC$ and $\muPL$ which are closely
related to $\mathsf{free}$ and $\mathsf{cofree}$. Given a pointed automaton
$(X,\overline{x},\delta)$,
$\nuC(X,\overline{x},\delta)$ corresponds to the largest set of equations
satisfied by the reachable part of $X$. On the other hand, for
an accepting automaton $(X,\delta,c)$, $\muPL(X,\delta,c)$ corresponds to the
smallest preformation of languages including $\{L(x,c)\mid x\in X\}$. Both functors
can be extended to pointed and accepting automata, making them endofunctors
on $\text{DA}$. Moreover, we show that $\nuC$ is an (idempotent) comonad and
$\muPL$ a monad.

We obtain the functors $\nuC$ and $\muPL$ by combining the
reachability-inclusion adjunction\footnote{the reachable functor sends a pointed automaton to its reachable part}, the transition-machine adjunction and
the lifted contravariant powerset adjunction \cite{bonchi:2014:algebraCoalgebra,bezhanishvili:2020:minimisation}.

\begin{center}
\begin{tikzpicture}[modal]
\node[] at (0,0) (0) [label=above:{}] {$\text{Alg}(G_1)$};
\node[] at (3,0) (1) [label=above:{}] {$\text{Alg}_r(G_1)$};
\node[] at (6,0) (2) [label=above:{}] {$\mathcal{C}$\phantom{CC}};
\node[rotate=90] at (1.5,0) {$\dashv$};
\node[rotate=90] at (4.5,0) {$\dashv$};
\node at (3,2.4) {$\nuC: \text{Alg}(G_1)\to \text{Alg}(G_1)$};
\node at (3,1.8) {$\nuC= I\circ M\circ T\circ R$};
\node[] at (8.5,0) (5) [label=above:{}] {$\text{CoAlg}(F_2)$};
\node[] at (11.5,0) (6) [label=above:{}] {$\text{Alg}(G_1)^{\text{op}}$};
\node[] at (14.5,0) (7) [label=above:{}] {$\mathcal{C}^{\text{op}}$};
\node[rotate=90] at (10,0) {$\vdash$};
\node[rotate=90] at (13,0) {$\vdash$};
\node at (11.5,2.4) {$\muPL: \text{CoAlg}(F_2)\to \text{CoAlg}(F_2)$};
\node at (11.5,1.8) {$\muPL= \overline{P}^{\text{op}}\circ \nuC^{\text{op}}\circ \overline{P}$};
\path[->] (0) edge[bend left] node[]{$R$} (1);
\path[->] (1) edge[bend left] node[]{$T$} (2);
\path[->] (1) edge[bend left] node[]{$I$} (0);
\path[->] (2) edge[bend left] node[]{$M$} (1);
\path[->] (5) edge[bend left] node[]{$\widehat{P}$} (6);
\path[->] (6) edge[bend left] node[]{$(T\circ R)^{\text{op}}$} (7);
\path[->] (6) edge[bend left] node[]{$\widehat{P}^{\text{op}}$} (5);
\path[->] (7) edge[bend left] node[]{$(I\circ M)^{\text{op}}$} (6);
\end{tikzpicture}
\end{center}

It is immediately clear from the diagram that $\nuC$ is a comonad and
$\muPL$ a monad. We have already seen that the unit of the transition-machine
Galois connection is pointwise the identity. The unit for the adjunction which
gives rise to $\nuC$ is actually the same, and so in particular pointwise
the identity. This means that $\mathcal{C}$ is a full coreflective subcategory
of $\text{Alg}(G_1)$ (and by extension also $\text{DA}$), and $\nuC$ is
an idempotent comonad. We depict the
monad and comonad situation in the following diagram:

\begin{center}
\begin{tikzpicture}[modal]
\node[] at (2,0) (0) [label=above:{}] {$(\Sigma^\ast)^\Sigma$};
\node[] at (5,0) (2) [label=above:{}] {$\nuC(\overline{x},\delta)^\Sigma$};
\node[] at (7.5,0) (3) [label=above:{}] {$X^\Sigma$};
\node[] at (10,0) (4) [label=above:{}] {$\muPL(\delta,c)^\Sigma$};
\node[] at (13,0) (6) [label=above:{}] {$(2^{\Sigma^\ast})^\Sigma$};
\node[] at (2,1.5) (7) [label=above:{}] {$\Sigma^\ast$};
\node[] at (5,1.5) (9) [label=above:{}] {$\nuC(\overline{x},\delta)$};
\node[] at (7.5,1.5) (a) [label=above:{}] {$X$};
\node[] at (10,1.5) (b) [label=above:{}] {$\muPL(\delta,c)$};
\node[] at (13,1.5) (d) [label=above:{}] {$2^{\Sigma^\ast}$};
\node[] at (2,3) (e) [label=above:{}] {$1$};
\node[] at (13,3) (f) [label=above:{}] {$2$};
\path[->] (0) edge node[]{} (2);
\path[->] (2) edge node[]{} (3);
\path[->] (3) edge node[]{} (4);
\path[->] (4) edge node[]{} (6);
\path[->] (7) edge node[]{} (0);
\path[->] (7) edge node[]{} (9);
\path[->] (9) edge node[]{} (2);
\path[->] (9) edge node[below]{$\epsilon_{(\overline{x},\delta)}$} (a);
\path[->] (a) edge node[]{$\delta$} (3);
\path[->] (a) edge node[below]{$\eta_{(\delta,c)}$} (b);
\path[->] (a) edge[bend left,looseness=0.5] node[]{$c$} (f);
\path[->] (b) edge node[]{} (4);
\path[->] (b) edge node[]{} (d);
\path[->] (b) edge[bend left,looseness=0.5] node[]{} (f);
\path[->] (d) edge node[]{} (6);
\path[->] (d) edge node[]{} (f);
\path[->] (e) edge node[]{} (7);
\path[->] (e) edge[bend left,looseness=0.5] node[]{} (9);
\path[->] (e) edge[bend left,looseness=0.5] node[]{$\overline{x}$} (a);
\end{tikzpicture}

\begin{align*}
  \epsilon_{(\overline{x},\delta)}: \nuC (\overline{x},\delta) &\longrightarrow (\overline{x},\delta) & \eta_{(\delta,c)}: (\delta,c) &\longrightarrow \muPL (\delta,c) \\
  [u] &\longmapsto \epsilon_{(\overline{x},\delta)}([u]) = \delta(\overline{x})(u) & x &\longmapsto \eta_{(\delta,c)}(x) = \{[u]\mid \delta(x)(u^r)\in c\} 
\end{align*}

\end{center}

In order to give some starting point on what $\nuC$ and $\muPL$ do, we concretely
specify them for $(X,\overline{x},\delta,c)$ in the table below. The entries
in gray correspond to what one gets upon lifting the functors to $\text{DA}$.
These can also neatly be seen in the diagram above; for instance, in order
to equip $\nuC(\overline{x},\delta)$ with accepting states one composes
$c$ with $\epsilon_{(\overline{x},\delta)}$. In a similar fashion, one can
equip $\muPL\left( \delta,c \right)$ with an initial state by composing
$\eta_{(\delta,c)}$ with $\overline{x}$.

\renewcommand{\arraystretch}{1.8}
\begin{table}[H]
  \centering
  \label{tab:none}
  \begin{tabular}{|c|c|c|}
  \hline
   & $\nuC(\overline{x},\delta)$ & $\muPL(\delta,c)$ \\ \hline
    state space & $\Sigma^\ast/{\left(\ker (\restr{\delta}{\left<\overline{x} \right>})^\sharp\right)}$ & $P\left(\Sigma^\ast/{\left( \ker \left(\restr{\widehat{\delta}}{\left<c \right>}\right)^\sharp \right) }\right)$\tablefootnote{here $\left<c \right>$ consists of all reachable states w.r.t. $\widehat{\delta}$} \\ \hline
  transition & $\sigma([w])(u)=[wu]$ & $\widehat{\sigma}(U)(u)=\{[w]\mid [wu^r]\in U\} $  \\ \hline
  initial state & $[\epsilon]$ & \textcolor{gray}{$\{[w]\mid \delta(\overline{x})(w^r)\in c\} $} \\ \hline
  final state & \textcolor{gray}{$\{[w]\mid \delta(\overline{x})(w)\in c\} $} & $\{U \mid [\epsilon]\in U\} $ \\ \hline
  \end{tabular}
\end{table}

The functor $\nuC$ is very closely related to $\mathsf{free}$. In fact, if
$X$ is reachable from an initial state $\overline{x}$, then
$\mathsf{free}(X,\delta)\cong \nuC(\overline{x},\delta)$ as
$G_1$-algebras. Additionally, even if our automaton consists of multiple
parts, we can obtain $\mathsf{free}(X,\delta)$ by gluing together
the $\nuC(x,\delta)$.

\begin{prop}
  Let $(X,\overline{x},\delta)$ be a pointed automaton.
  \begin{enumerate}
    \item If $X=\left<\overline{x} \right>$, then $\mathsf{free}(X,\delta)\cong \nuC(\overline{x},\delta)$,
    \item $\mathsf{free}(X,\delta)\cong \Pi_{x\in X}\nuC(x,\delta)$ in $\text{Alg}_r(G_1)$.
  \end{enumerate}
\end{prop}

\begin{proof}
  For the first point, note that if $X$ is reachable, then $\nuC(\overline{x},\delta)$ 
  is just the transition monoid of $X$, which is isomorphic to 
  $\mathsf{free}(X,\delta)$.

  For the second point, we have a map $\mathsf{free}(X,\delta)\to \nuC(x_i,\delta)$ 
  for each $x_i\in X$, given by $[u]_{\sim} \mapsto [u]_{\sim_i}$
  where we view $\mathsf{free}(X,\delta)$ as the transition monoid of  $X$ 
  presented as a quotient congruence $\sim$, and where
  $\sim_i=\ker (\restr{\delta}{\left<x_i \right>})^\sharp$. Hence
  we get a map $\mathsf{free}(X,\delta)\to \Pi_{x\in X}\nuC(x,\delta)$
  in $\text{Alg}_r(G_1)$. The map in the other direction is constructed as
  follows. Let $\sim_i = \ker (\restr{\delta}{\left<x_i \right>})^\sharp$.
  We build the congruence $\sim'$ by $u\sim' v \iff \forall i: u\sim_i v$, and
  claim that $\sim'=\sim$. Clearly we have ${\sim}\subseteq {\sim'}$. Let $u\sim' v$, then
  $u\sim_i v$ for all $i$. Take some arbitrary $x_i\in X$, then $u\sim_i v$ so
  $\delta(x_i,u)=\delta(x_i,v)$, which means that $u\sim v$. As we are in
  $\text{Alg}_r(G_1)$, the product automaton is reachable and any
  tuple is of the shape $([u]_{\sim_i})_{i}$. We then map
  such a tuple to $[u]_\sim$ which is well-defined by the above argument.
\end{proof}

Before we can establish a relationship between $\mathsf{cofree}$ and
$\muPL$, we have to establish some useful properties which allow us to
characterise $\muPL(\delta,c)$ as the least preformation of languages
containing $L(x,\delta,c)$ for all $x\in X$.

\begin{prop}
  For an accepting automaton $(X,\delta,c)$, $\muPL(\delta,c)$ is minimal.
\end{prop}

\begin{proof}
  Let $U$ be a state in $\muPL(\delta,c)$. Then
  \[
    u\in L(U) \iff [\epsilon]\in \widehat{\sigma}(U)(u)\iff [\epsilon]\in \{[w]\mid [wu^r]\in U\} \iff [u^r]\in U.
  \]
  So $L(U)=\{u\mid [u^r]\in U\}$ and $L(U)=L(V)\implies U=V$. Hence any two
  bisimilar states are equal, and $\muPL(\delta,c)$ is minimal.
\end{proof}

\begin{prop}
  For an accepting automaton $(X,\delta,c)$, $\muPL(\delta,c)$ is (isomorphic
  to) a preformation of languages.
\end{prop}

\begin{proof}
  From the definition of $\muPL$ and the previous proposition, it is clear
  that $\muPL$ is a complete atomic Boolean algebra and isomorphic to a
  complete atomic Boolean algebra of languages with
  $L(U)\cup L(V)=L(U\cup V)$,  $L(U)\cap L(V)=L(U\cap V)$ 
  and $\overline{L(U)}=L(\overline{U})$. Moreover, it has a right
  derivative given by the transition function defined on it. It is easy to
  see that one can also define a left derivative on it as the generators
  (atoms) are congruences. Both derivatives coincide with the standard derivatives
  on languages. Hence $\muPL(\delta,c)$ is (isomorphic to)
  a preformation of languages.
\end{proof}

In light of this, we may think of the elements of $\muPL(\delta,c)$ as languages.
In particular, if we say that a language belongs to $\muPL(\delta,c)$, it means
that there exists a state in $\muPL(\delta,c)$ whose language is the same.

\begin{prop}
  Let $(X,\delta,c)$ be an accepting automaton. For any $x\in X$,
  $L(x,c)$ belongs to $\muPL(\delta,c)$.
\end{prop}

\begin{proof}
  This follows simply from the existence of the unit $\mu_{(\delta,c)}:X\to \muPL(\delta,c)$.
\end{proof}

\begin{prop}\label{prop:muPLSmallestPL}
  Let $(X,\delta,c)$ be an accepting automaton. Then
  $\muPL(\delta,c)$ is (isomorphic to) the smallest preformation of languages
  including $\{L(x,c)\mid x\in X\}$.
\end{prop}

\begin{proof}
  We have already established that $\muPL(\delta,c)$ is a preformation of
  languages which includes $L(x,c)$ for any $x\in X$.
  We wish to show that $\muPL(\delta,c)$ is the minimal preformation
  of languages having this property. To do this we first note that if
  $L(x,c)$ belongs to a preformation of languages, so does $\overline{L(x,c)}$
  and $L(x,\widehat{\delta}(c)(a))$. The first follows as a preformation of
  languages is closed under complement, and the second as it is closed under
  right and left derivatives. In particular, one has
  $u\in L(x,\widehat{\delta}(c)(a)) \iff ua\in L(x,c) \iff u\in \prescript{}{a}{L(x,c)}$.
  Our goal is to show that any atom in $\muPL(\delta,c)$ can be written as a Boolean
  combination of such languages. Then any preformation of languages satisfying
  the requirements from the proposition must include the atoms of
  $\muPL(\delta,c)$ and hence completely embed it.
  We recall that the atoms are of the shape $\{[w]\}$ where the equivalence
  classes were generated by $\ker \left(\restr{\widehat{\delta}}{\left<c \right>}\right)^\sharp$.
  Let $u\in \Sigma^\ast$. Then
  \allowdisplaybreaks{
  \begin{align*}
    u\in L(\{[w]\}) &\iff [u^r]=[w] \\
                    &\iff \forall U\in \left<c\right>: \widehat{\delta}(U)(u^r)=\widehat{\delta}(U)(w)\\
                    &\iff \forall x\in X,U\in \left<c\right>: \delta(x,u)\in U \iff \delta(x,w^r)\in U\\
                    &\iff \forall x\in X: \left(\forall\ \delta(x,w^r)\in U\in \left<c \right>: \delta(x,u)\in U\right)\land \left(\forall\ \delta(x,w^r)\not\in U\in \left<c \right>: \delta(x,u)\not\in U\right) \\
                    &\iff \forall x\in X: u\in \bigcap_{\substack{U\in \left<c \right>\\ \delta(x,w^r)\in U}} L(x,U) \text{ and } u\not\in \bigcup_{\substack{U\in \left<c \right>\\ \delta(x,w^r)\not\in U}} L(x,U)\\
                    &\iff \forall x\in X: u\in \bigcap_{\substack{U\in \left<c \right>\\ \delta(x,w^r)\in U}} L(x,U) \text{ and } u\in \bigcap_{\substack{U\in \left<c \right>\\ \delta(x,w^r)\not\in U}} \overline{L(x,U)}\\
                    &\iff u\in \bigcap_{x\in X}\left(\bigcap_{\substack{U\in \left<c \right>\\ \delta(x,w^r)\in U}} L(x,U) \cap \bigcap_{\substack{U\in \left<c \right>\\ \delta(x,w^r)\not\in U}} \overline{L(x,U)}\right) 
  \end{align*}
}
  This concludes the proof.
\end{proof}

We are now in a position where we can make the link between $\muPL$ and
$\mathsf{cofree}$.

\begin{cor}\label{cor:embedCofree}
  Let $(X,\delta,c)$ be an accepting automaton such that the smallest Boolean
  subalgebra of $P(X)$ containing $\left<c \right>$ is $P(X)$ itelf.
  Then $\mathsf{cofree}(X,\delta)$ can entirely be embedded in
  $\muPL(\delta,c)$.
\end{cor}

\begin{proof}
  If the assumption holds, any subset $U\subseteq X$ is a Boolean combination of subsets
  reachable from $c$ via $\widehat{\delta}$. This means that $\muPL(\delta,c)$
  contains $L(x,U)$ for any $x\in X$ and any $U\subseteq X$, so in particular it
  contains any language in $\mathsf{cofree}(X,\delta)$.
\end{proof}

\subsection{An Example}

Before concluding the section, we compute $\nuC$ and $\muPL$ for the concrete
DFA shown below (which can also be found in \cite{ballester:2015:dualEquivalenceOfEquationsAndCoequations}).

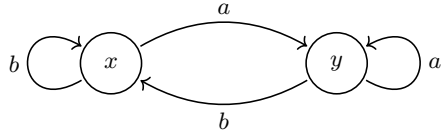
\begin{center}
\begin{tikzpicture}[modal]
\node[state] at (0,0) (1) [label=above:{}] {$x$};
\node[state] at (3,0) (2) [label=above:{}] {$y$};
\path[->] (1) edge[bend left] node[]{$a$} (2);
\path[->] (1) edge[reflexive left] node[]{$b$} (1);
\path[->] (2) edge[bend left] node[]{$b$} (1);
\path[->] (2) edge[reflexive right] node[right]{$a$} (2);
\end{tikzpicture}
\end{center}

As the automaton is reachable regardless of the choice of initial state,
the automata $\nuC(x,\delta)$, $\nuC(y,\delta)$ and $\mathsf{free}(\delta)$ 
are all isomorphic. We depict $\nuC(x,\delta)$ below
on the left ($\sim=\ker \delta^\sharp$) and $\nuC(x,\delta,\{x\})$ on the
right.

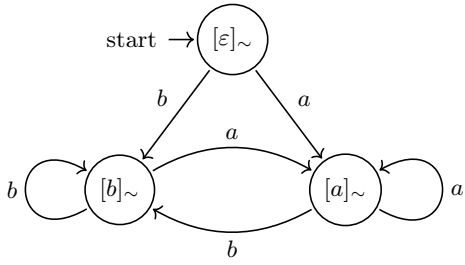
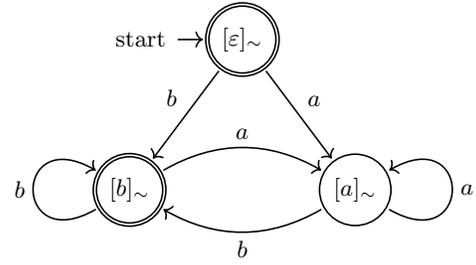
\begin{center}
\begin{tikzpicture}[modal]
  \node[initial,state] at (1.5,2) (0) [label=above:{}] {$[\epsilon]_\sim$};
  \node[state] at (0,0) (1) [label=above:{}] {$[b]_\sim$};
  \node[state] at (3,0) (2) [label=above:{}] {$[a]_\sim$};
\path[->] (0) edge[] node[above left]{$b$} (1);
\path[->] (0) edge[] node[]{$a$} (2);
\path[->] (1) edge[bend left] node[]{$a$} (2);
\path[->] (1) edge[reflexive left] node[]{$b$} (1);
\path[->] (2) edge[bend left] node[]{$b$} (1);
\path[->] (2) edge[reflexive right] node[right]{$a$} (2);

  \node[initial,state,accepting] at (9.5,2) (0b) [label=above:{}] {$[\epsilon]_\sim$};
  \node[state,accepting] at (8,0) (1b) [label=above:{}] {$[b]_\sim$};
  \node[state] at (11,0) (2b) [label=above:{}] {$[a]_\sim$};
\path[->] (0b) edge[] node[above left]{$b$} (1b);
\path[->] (0b) edge[] node[]{$a$} (2b);
\path[->] (1b) edge[bend left] node[]{$a$} (2b);
\path[->] (1b) edge[reflexive left] node[]{$b$} (1b);
\path[->] (2b) edge[bend left] node[]{$b$} (1b);
\path[->] (2b) edge[reflexive right] node[right]{$a$} (2b);
\end{tikzpicture}
\end{center}

Next we compute $\muPL(\delta,\{x\})$. This can be done in several stages.
First we apply the lifted powerset functor to our accepting automaton, which
results in the pointed automaton to the left. After applying $\nuC$ to this
automaton we obtain the automaton on the right ($\approx =\ker \widehat{\delta}^\sharp$).

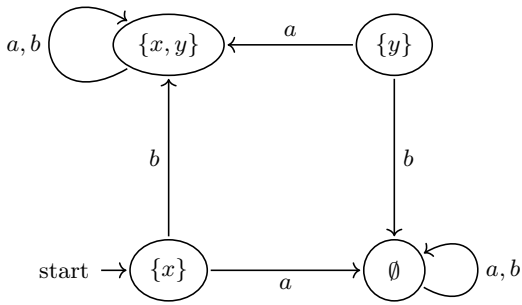
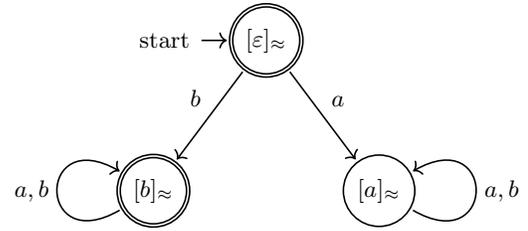
\begin{center}
\begin{tikzpicture}[modal]
\node[state,initial,ellipse] at (0,0) (0) [label=above:{}] {$\{x\} $};
\node[state,ellipse] at (3,0) (1) [label=above:{}] {$\emptyset$};
\node[state,ellipse] at (0,3) (2) [label=above:{}] {$\{x,y\} $};
\node[state,ellipse] at (3,3) (3) [label=above:{}] {$\{y\} $};
\path[->] (0) edge node[below]{$a$} (1);
\path[->] (0) edge node[]{$b$} (2);
\path[->] (1) edge[reflexive right] node[right]{$a,b$} (1);
\path[->] (2) edge[reflexive left] node[]{$a,b$} (2);
\path[->] (3) edge node[above]{$a$} (2);
\path[->] (3) edge node[]{$b$} (1);

\node[initial,state,accepting] at (9.5,2.5) (0b) [label=above:{}] {$[\epsilon]_{\approx}$};
  \node[state,accepting] at (8,0.5) (1b) [label=above:{}] {$[b]_{\approx}$};
  \node[state] at (11,0.5) (2b) [label=above:{}] {$[a]_{\approx}$};
\path[->] (0b) edge[] node[above left]{$b$} (1b);
\path[->] (0b) edge[] node[]{$a$} (2b);
\path[->] (1b) edge[reflexive left] node[]{$a,b$} (1b);
\path[->] (2b) edge[reflexive right] node[right]{$a,b$} (2b);

\end{tikzpicture}
\end{center}

Finally, we apply the lifted contravariant powerset functor once more
to obtain $\muPL(\delta,\{x\})$ (on the left) and the corresponding
preformation of languages (on the right).

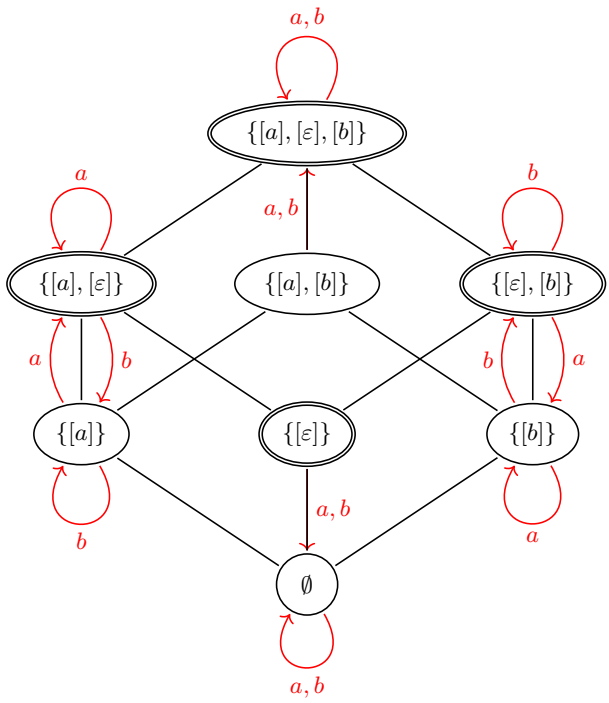
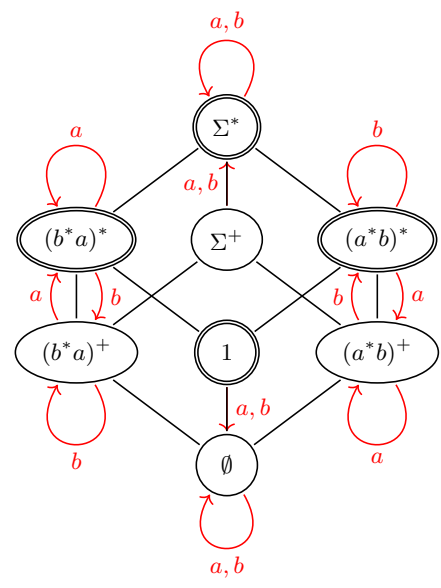
\begin{center}
\begin{tikzpicture}[modal]
\node[state,ellipse] at (3,0) (0) [label=above:{}] {$\emptyset$};
\node[state,ellipse] at (0,2) (1) [label=above:{}] {$\{[a]\} $};
\node[state,accepting,ellipse] at (3,2) (2) [label=above:{}] {$\{[\epsilon]\} $};
\node[state,ellipse] at (6,2) (3) [label=above:{}] {$\{[b]\} $};
\node[state,accepting,ellipse] at (0,4) (4) [label=above:{}] {$\{[a],[\epsilon]\} $};
\node[state,ellipse] at (3,4) (5) [label=above:{}] {$\{[a],[b]\} $};
\node[state,accepting,ellipse] at (6,4) (6) [label=above:{}] {$\{[\epsilon],[b]\} $};
\node[state,accepting,ellipse] at (3,6) (7) [label=above:{}] {$\{[a],[\epsilon],[b]\} $};
\path[->,red] (0) edge[reflexive below] node[]{$a,b$} (0);
\path[->,red] (1) edge[bend left] node[]{$a$} (4);
\path[->,red] (1) edge[reflexive below] node[]{$b$} (1);
\path[->,red] (2) edge node[]{$a,b$} (0);
\path[->,red] (3) edge[reflexive below] node[below]{$a$} (3);
\path[->,red] (3) edge[bend left] node[]{$b$} (6);
\path[->,red] (4) edge[reflexive above] node[above]{$a$} (4);
\path[->,red] (4) edge[bend left] node[]{$b$} (1);
\path[->,red] (5) edge node[]{$a,b$} (7);
\path[->,red] (6) edge[bend left] node[]{$a$} (3);
\path[->,red] (6) edge[reflexive above] node[above]{$b$} (6);
\path[->,red] (7) edge[reflexive above] node[above]{$a,b$} (7);

\path[-] (0) edge (1);
\path[-] (0) edge (2);
\path[-] (0) edge (3);
\path[-] (1) edge (4);
\path[-] (1) edge (5);
\path[-] (2) edge (4);
\path[-] (2) edge (6);
\path[-] (3) edge (5);
\path[-] (3) edge (6);
\path[-] (4) edge (7);
\path[-] (5) edge (7);
\path[-] (6) edge (7);

\node[state,ellipse] at (11,0.75) (0b) [label=above:{}] {$\emptyset$};
\node[state,ellipse] at (9,2.25) (1b) [label=above:{}] {$(b^\ast a)^+$};
\node[state,accepting,ellipse] at (11,2.25) (2b) [label=above:{}] {$1$};
\node[state,ellipse] at (13,2.25) (3b) [label=above:{}] {$(a^\ast b)^+$};
\node[state,accepting,ellipse] at (9,3.75) (4b) [label=above:{}] {$(b^\ast a)^\ast$};
\node[state,ellipse] at (11,3.75) (5b) [label=above:{}] {$\Sigma^+$};
\node[state,accepting,ellipse] at (13,3.75) (6b) [label=above:{}] {$(a^\ast b)^\ast$};
\node[state,accepting,ellipse] at (11,5.25) (7b) [label=above:{}] {$\Sigma^\ast$};
\path[->,red] (0b) edge[reflexive below] node[]{$a,b$} (0b);
\path[->,red] (1b) edge[bend left] node[]{$a$} (4b);
\path[->,red] (1b) edge[reflexive below] node[]{$b$} (1b);
\path[->,red] (2b) edge node[]{$a,b$} (0b);
\path[->,red] (3b) edge[reflexive below] node[below]{$a$} (3b);
\path[->,red] (3b) edge[bend left] node[]{$b$} (6b);
\path[->,red] (4b) edge[reflexive above] node[above]{$a$} (4b);
\path[->,red] (4b) edge[bend left] node[]{$b$} (1b);
\path[->,red] (5b) edge node[]{$a,b$} (7b);
\path[->,red] (6b) edge[bend left] node[]{$a$} (3b);
\path[->,red] (6b) edge[reflexive above] node[above]{$b$} (6b);
\path[->,red] (7b) edge[reflexive above] node[above]{$a,b$} (7b);

\path[-] (0b) edge (1b);
\path[-] (0b) edge (2b);
\path[-] (0b) edge (3b);
\path[-] (1b) edge (4b);
\path[-] (1b) edge (5b);
\path[-] (2b) edge (4b);
\path[-] (2b) edge (6b);
\path[-] (3b) edge (5b);
\path[-] (3b) edge (6b);
\path[-] (4b) edge (7b);
\path[-] (5b) edge (7b);
\path[-] (6b) edge (7b);
\end{tikzpicture}
\end{center}

Interestingly, Corollary \ref{cor:embedCofree} applies to $\muPL(\delta,c)$, hence
$\mathsf{cofree}(\delta)$ is embedded in the preformation of languages
given above ($\mathsf{cofree}(\delta)$ has state space $\{(b^\ast a)^+,(b^\ast a)^\ast,(a^\ast b)^+,(a^\ast b)^\ast\}$).
We also illustrate how to compute the atom $\{[\epsilon]\}$ from the languages
in $\mathsf{cofree}(\delta)$ as in the proof of Proposition \ref{prop:muPLSmallestPL}.
We have that
\begin{align*}
  L(\{[\epsilon]\}) &= \bigcap_{x\in X}\left(\bigcap_{\substack{U\subseteq X\\ x\in U}} L(x,U) \cap \bigcap_{\substack{U\subseteq X\\ x\not\in U}} \overline{L(x,U)}\right) \\
                    &= \left(L(x,\{x\})\cap L(x,\{x,y\} \cap \overline{L(x,\emptyset)} \cap \overline{L(x,\{y\})}\right)\cap \left(L(y,\{y\})\cap L(y,\{x,y\} \cap \overline{L(y,\emptyset)} \cap \overline{L(y,\{x\})}\right)\\
                    &= L(x,\{x\})\cap L(y,\{y\}).
\end{align*}

\section{Instantiating to Lasso Automata}\label{section:LassoAutomata}

After having established a clear picture for deterministic automata, our goal
of this section is to do the same for lasso automata and the closely related
$\Omega$-automata. Lasso automata are very similar to DFAs, in fact,
one may view them as a DFA with an extended alphabet \cite{calbrix:1994:ultimatelyPeriodicWords}. They operate on lassos,
that is, pairs of words  $(u,v)\in\Sigma^\ast\times\Sigma^+$.

Our primary interest in studying these structures is their close relationship
to $\Omega$-automata, which are lasso automata that have some additional
structural properties, allowing them to be used as acceptors of $\omega$-languages.
Given an ultimately periodic word $uv^\omega$, we check whether it is accepted
by an $\Omega$-automata, by checking whether it accepts the lasso $(u,v)$.
This is well-defined for $\Omega$-automata, as they don't make a difference
(in terms of acceptance) between lassos which represent the same ultimately
periodic word.

Through the transition-machine construction, we hope to gain further insights
into regular $\omega$-languages, and establish closer links between the
algebraic and coalgebraic study of $\omega$-languages.

Before we begin our constructions, we briefly introduce the picture we
are working in (these facts can be found in \cite{cruchten:2022:omegaAutomata}).
As was mentioned in the preliminaries, we stick to the same
names for the functors as we try to view the categorical picture as
a framework:

\begin{align*}
  G(X_1,X_2)&=(X_1\times \Sigma,X_1\times\Sigma+ X_2\times\Sigma) & F(X_1,X_2)&=(X_1^\Sigma\times X_2^\Sigma,X_2^\Sigma), \\
  G_1(X_1,X_2)&=(1,\emptyset)+G(X_1,X_2) & F_2(X_1,X_2)&=F(X_1,X_2)\times (1,2).
\end{align*}

As in the DFA setting, $\text{Alg}(G)$ is isomorphic to
$\text{CoAlg}(F)$ and $G \dashv F$.
The initial $G_1$-algebra has carrier $(\Sigma^\ast,\Sigma^{\ast,+})$, initial
state $\epsilon$ and the three transitions are given by
\[
  \sigma_1(u,a)= ua, \qquad \sigma_2(u,a)= (u,a), \qquad \sigma_3((u,v), a)= (u,va).
\] 
The terminal $F_2$-coalgebra has carrier $(2^{\Sigma^{\ast +}},2^{\Sigma^\ast})$,
final states $\{U\subseteq \Sigma^\ast\mid \epsilon\in U\}$ and transitions
(again given in order):
\[
  L \mapsto \lambda a. \{(u,v)\in\Sigma^{\ast +}\mid (au,v)\in L\}, \qquad  L \mapsto \lambda a. \{u \in\Sigma^{\ast}\mid (\epsilon,au)\in L\}, \qquad U \mapsto \lambda a. \{u\in\Sigma^{\ast}\mid au\in U\}.
\] 
We end up with the following setup.

\begin{center}
\begin{tikzpicture}[modal]
\node[] at (2,0) (0) [label=above:{}] {$((\Sigma^\ast)^\Sigma\times (\Sigma^{\ast +})^\Sigma, (\Sigma^{\ast +})^\Sigma)$};
\node[] at (7.5,0) (3) [label=above:{}] {$(X_1^\Sigma\times X_2^\Sigma,X_2^\Sigma)$};
\node[] at (13,0) (6) [label=above:{}] {$((2^{\Sigma^{\ast +}})^\Sigma\times (2^{\Sigma^\ast})^\Sigma, (2^{\Sigma^\ast})^\Sigma)$};
\node[] at (2,1.5) (7) [label=above:{}] {$(\Sigma^\ast, \Sigma^{\ast +})$};
\node[] at (7.5,1.5) (a) [label=above:{}] {$(X_1,X_2)$};
\node[] at (13,1.5) (d) [label=above:{}] {$(2^{\Sigma^{\ast +}}, 2^{\Sigma^\ast})$};
\node[] at (2,3) (e) [label=above:{}] {$(1,0)$};
\node[] at (13,3) (f) [label=above:{}] {$(1,2)$};
\path[->] (0) edge node[]{} (3);
\path[->] (3) edge node[]{} (6);
\path[->] (7) edge node[]{} (0);
\path[->] (7) edge node[]{} (a);
\path[->] (a) edge node[]{$\delta$} (3);
\path[->] (a) edge node[below]{} (d);
\path[->] (a) edge[bend left,looseness=0.5] node[]{$c$} (f);
\path[->] (d) edge node[]{} (6);
\path[->] (d) edge node[]{} (f);
\path[->] (e) edge node[]{} (7);
\path[->] (e) edge[bend left,looseness=0.5] node[]{$\overline{x}$} (a);
\end{tikzpicture}
\end{center}

Similarly to \cite{ballester:2015:dualEquivalenceOfEquationsAndCoequations}, one
can now define equations and coequations, and
for each lasso automaton $\mathsf{free}$ and $\mathsf{cofree}$. As we show later,
these notions don't just appear natural, but they behave as one would
suspect. In particular, we show that $\mathsf{Eq}(\left<L \right>)$ 
corresponds to the syntactic congruence of $L$.

\begin{defn}[\cite{ballester:2015:dualEquivalenceOfEquationsAndCoequations}]
  A \emph{set of equations} is a bisimulation equivalence relation
  $E=(E_1,E_2)\subseteq (\Sigma^\ast,\Sigma^{\ast +})$ on the lasso automaton
  $(X_1,X_2,\delta_1,\delta_2,\delta_3)$. For $x\in X_1$, we define
  $(X,x) \models E$ (the pointed lasso automaton satisfies $E$) by
  \[
    (X,x) \models E \iff \forall u,v\in E_1: \delta_1(x,u)=\delta_1(x,v) \text{ and }\forall (u,v),(u',v')\in E_2: \delta(x,(u,v))=\delta(x,(u',v')).
  \] 
  The lasso automaton satisfies $E$, $X\models E$, iff it satisfies $E$ for all $x\in X_1$.
\end{defn}

\begin{defn}[\cite{ballester:2015:dualEquivalenceOfEquationsAndCoequations}]
  A \emph{set of coequations} is a subautomaton $D=(D_1,D_2)\subseteq (2^{\Sigma^{\ast +}},2^{\Sigma^\ast})$.
  We define $(X,c)\models D$ (the accepting automaton satisfies $D$) by
  \[
    (X,c)\models D \iff \forall x\in X_1: L(x)\in D_1.
  \] 
  The lasso automaton satisfies $D$ iff $(X,c)\models D$ for all $c\subseteq X_2$.
\end{defn}

With this terminology in place, we also get a largest set of equations
satisfied by a lasso automaton $(X_1,X_2,\delta_1,\delta_2,\delta_3)$ which
we denote by $\mathsf{Eq}(X)$ and dually also a smallest set of coequations which
we denote $\mathsf{CoEq}(X)$. We do not provide specific constructions for
$\mathsf{free}$ and $\mathsf{cofree}$ in this paper.

\subsection{Transition and Machine Constructions Anew}

Due to how $\mathsf{Eq}(X)$ is defined, $\mathsf{free}(X)$ comes with certain
algebraic structure. In this section, we look at what that structure is
and give a transition and machine construction, which forms
an adjunction. In that way, we obtain two functors $\nuC$ and $\muPL$, 
which respectively form a comonad and a monad. These constructions are not
very involved and our methodology remains unchanged.

We start by defining some very minimal and natural structure on the
carrier of the initial algebra. Let
$\cdot:\Sigma^\ast\times \Sigma^\ast \to \Sigma^\ast$ be concatenation
of words and $\times: \Sigma^\ast \times\Sigma^{\ast +}\to \Sigma^{\ast +}$ 
be given by $u\times (v,w)=(uv,w)$. We often write $\cdot$ for $\times$ when
this does not lead to confusion. Any largest set of equations respects these two
operations.

\begin{prop}
  For any lasso automaton $X$, $\mathsf{Eq}\left( X \right)$ is a congruence (in the above sense) on $(\Sigma^\ast,\Sigma^{\ast +})$.
\end{prop}

\begin{proof}
  The proof consists mainly of unravelling definitions so we omit it.
\end{proof}

We show later that $\mathsf{Eq}(\left<L \right>)$ is the syntactic
congruence of $L$. By that we mean that $(\Sigma^{\ast},\Sigma^{\ast +})/{\mathsf{Eq}(\left<L \right>)}$ 
is isomorphic to $\left<L \right>$ and hence also minimal. This recovers the
classical situation where the minimal DFA corresponds to the syntactic
monoid and vice versa.

We move on to the transition and machine constructions. The categories
involved are that of reachable $G_1$-algebras and bisimulation congruences
over $(\Sigma^\ast,\Sigma^{\ast +})$ (which we just write $\mathcal{C}$).
By bisimulation congruence we mean a congruence which is also a bisimulation
equivalence on the automaton given by $(\Sigma^\ast,\Sigma^{\ast +})$.
We remark, that a congruence on the free  $G_1$-algebra (where we specifically
talk about congruences on algebras for an endofunctor), is also a bisimulation.
The congruence we have defined above does not correspond precisely to this
as we extended multiplication to words to equip the carrier with a certain
structure.

Our first claim is that $\mathcal{C}$ and $\text{Alg}_r(G_1)$
are thin.

\begin{lem}\label{lem:thin2}
  The category $\text{Alg}_r(G_1)$ of reachable $G_1$ algebras is thin.
\end{lem}

\begin{proof}
  The claim is clear for $\mathcal{C}$, so we only show it for $\text{Alg}_r(G_1)$.
  Let $f,g:(X,\overline{x},\delta_{1,X},\delta_{2,X},\delta_{3,X})\to (Y,\overline{y},\delta_{1,Y},\delta_{2,Y},\delta_{3,Y})$, $x\in X_1$ and $x'\in X_2$. By reachability we can find
  $w\in \Sigma^\ast$ and $(u,v)\in \Sigma^{\ast +}$ such that
  $x=\delta_{1,X}(\overline{x})(w)$ and $x'=\delta_X(\overline{x})(u,v)$. Then
  \begin{align*}
    f_1(x)&=f_1(\delta_{1,X}(\overline{x})(w))=\delta_{1,Y}(\overline{y})(w)=\ldots = g_1(x) \\
    f_2(x')&=f_2(\delta_{X}(\overline{x})(u,v))=\delta_Y(\overline{y})(u,v)=\ldots =g_2(x').
  \end{align*}
  Hence there can only be at most one morphism between any two reachable
  pointed lasso automata and  $\text{Alg}_r(G_1)$ is thin.
\end{proof}

One may also check that morphisms in $\text{Alg}_r(G_1)$ seen as morphisms in
Set are always surjective. Next we define the transition and machine functors
as follows:

\vspace{1em}
\begin{minipage}[c][][c]{\textwidth}
  \hspace{-0.5cm}
  \begin{minipage}{0.40\textwidth}
    \[
    T : \text{Alg}_r(G_1)\to \mathcal{C}
    \] 
    \begin{itemize}
      \item Objects: $(\overline{x},\delta_1,\delta_2,\delta_3) \mapsto (\ker \delta_1^\sharp,\ker \delta^\sharp)$
      \item Morphisms:
        \[
          h:(\overline{x},X)\twoheadrightarrow (\overline{y},Y) \implies T(\overline{x},X) \subseteq T(\overline{y},Y)
        \]
    \end{itemize}
  \end{minipage}
  \begin{minipage}{0.50\textwidth}
    \[
    M : \mathcal{C} \to \text{Alg}_r(G_1)
    \] 
    \begin{itemize}
      \item Objects: $(C_1,C_2) \mapsto ([\epsilon]_{C_1},\sigma_{1,C},\sigma_{2,C},\sigma_{3,C})$
        \begin{align*}
          \sigma_{1,C}([w]_{C_1})(a)&=[wa]_{C_1}\\
          \sigma_{2,C}([w]_{C_1})(a)&=[(w,a)]_{C_2}\\
          \sigma_{3,C}([(u,v)]_{C_2})(a)&=[(u,va)]_{C_2}
        \end{align*}
      \item Morphisms:
        \[
          h: (C_1,C_2) \subseteq (D_1,D_2) \mapsto Mh_i([x]_{C_i})=[x]_{D_i}
        \]
    \end{itemize}
  \end{minipage}
  \vspace{1em}
\end{minipage}

It is clear that the functor $M$ is well defined; the transition maps are
well defined as $(C_1,C_2)$ is by definition a bisimulation equivalance, and
the morphisms are well-defined by properties of congruences. It remains
to show that $T$ is well-defined which we show in the next lemma.

\begin{lem}
  The functor $T$ is well-defined.
\end{lem}

\begin{proof}
  It is clear that both $\ker \delta_1^\sharp$ and $\ker \delta^\sharp$ are
  equivalence relations.
  As $(\overline{x},\delta_1)$ is just a DFA, we already know that
  $\ker \delta_1^\sharp$ is a congruence over $\Sigma^\ast$. Let
  $((u,v),(u',v'))\in \ker \delta^\sharp$, so for each $x\in X_1$ we have
  that $\delta(x,(u,v))=\delta(x,(u',v'))$. Next let $(w,w')\in\ker \delta_1^\sharp$.
  For each $x\in X_1$ we have that
  \[
    \delta(x,(wu,v))=\delta(\delta_1(x,w),(u,v))=\delta(\delta_1(x,w'),(u',v'))=\delta(x,(w'u',v')).
  \] 
  Hence $(\ker \delta_1^\sharp,\ker \delta^\sharp)$ is a congruence over
  $(\Sigma^\ast,\Sigma^{\ast +})$. We also note that it is equal to
  $\mathsf{Eq}\left( X \right) $ (see next lemma) and in particular a bisimulation equivalence,
  which is easy to show (it follows directly from determinedness, if
  $\delta(x,(u,v))=\delta(x,(u',v'))$ then $\delta(x,(u,va))=\delta(x,(u',v'a))$).

  Next, we claim that if $h:(\overline{x},X)\twoheadrightarrow (\overline{y},Y)$
  then $(\ker \delta_{X,1}^\sharp,\ker \delta_{X}^\sharp) \subseteq (\ker \delta_{Y,1}^\sharp,\ker \delta_Y^\sharp)$.
  Again, $h_1$ is just a morphism between pointed automata, so we know that
  $\ker \delta_{X,1}^\sharp \subseteq \ker \delta_{Y,1}^\sharp$. For the rest
  of the claim let $((u,v),(u',v'))\in \ker \delta_X^\sharp$ and $y\in Y_1$.
  As $h$ is surjective, there exists some $x\in X_1$ such that
  $y=h_1(x)$. We now have that
  \[
    \delta_Y(y,(u,v))=h(\delta_X(x,(u,v)))=h(\delta_X(x,(u',v')))=\delta_Y(y,(u',v')).
  \] Hence $((u,v),(u',v'))\in\ker \delta_Y^\sharp$.
\end{proof}

\begin{lem}\label{lem:Eq}
  Let $(X,\overline{x})$ be a reachable pointed lasso automaton. Then
  $T(X,\overline{x})=\mathsf{Eq}(X)$.
\end{lem}

\begin{proof}
  As $\mathsf{Eq}\left(X \right)$ is the largest set of equations satisfied
  by $X$, we immediately have that $T(X,\overline{x})\subseteq \mathsf{Eq}(X)$.
  For the other direction, let $(u,v)\in \mathsf{Eq}(X)_1$, so for all
  $x\in X_1: \delta_1(x,u)=\delta_1(x,v)$, then by definition
  $(u,v)\in \ker \delta_1^\sharp = T(X,\overline{x})_1$. Finally, let
  $((u,v),(u',v'))\in \mathsf{Eq}(X)_2$, so for all $x\in X_1: \delta(x,(u,v))=\delta(x,(u',v'))$.
  Hence $((u,v),\left( u',v' \right))\in\ker \delta^\sharp = T(X,\overline{x})_2$.
\end{proof}

The last claim is that $T$ is a right adjoint of $M$. We again show that
the unit of this Galois connection is just the identity.

\begin{lem}
  For any bisimulation congruence $(C_1,C_2)$, we have that
  $(C_1,C_2) = MT(C_1,C_2)$.
\end{lem}

\begin{proof}
  First we look at the first sort. Let $u,v\in \Sigma^\ast$. Then
  \begin{align*}
    (u,v)\in MTC_1 &\iff \forall w\in \Sigma^\ast: \sigma_1([w],u)=\sigma_1([w],v) \\
    &\iff \forall w\in \Sigma^\ast: [wu]=[wv] \\
    &\iff [u]=[v] && \text{$w=\epsilon$ and as $C_1$ is a congruence} \\
    &\iff (u,v)\in C_1.
  \end{align*}
  Next, let $(u,v),(u',v')\in \Sigma^{\ast +}$. Then
  \begin{align*}
    ((u,v),(u',v'))\in MTC_2 &\iff \forall w\in \Sigma^\ast: \sigma([w],(u,v))=\sigma([w],(u',v')) \\
    &\iff \forall w\in \Sigma^\ast: [(wu,v)]=[(wu',v')] \\
    &\iff [(u,v)]=[(u',v')] && \text{$w=\epsilon$ and as $C_1$ is a congruence} \\
    &\iff ((u,v),(u',v'))\in C_2.
  \end{align*}
\end{proof}

The next lemma gives the counit of the Galois connection.

\begin{lem}
  Let $(X,\overline{x})$ be a pointed lasso automaton. Then
  $\epsilon_{(X,\overline{x})}:(MTX,[\epsilon])\to (X,\overline{x})$ given
  by $\epsilon_{1}([u])=\delta_1(x,u)$ and $\epsilon_2([(u,v)])=\delta(x,(u,v))$ 
  is a $G_1$-Algebra morphism.
\end{lem}

\begin{proof}
  This is well defined by how we obtained the equivalence classes. To show that
  it is a $G_1$-Algebra morphism we have that it preserves initial states by
  the definition of $\epsilon_1$. Moreover, it also respects transitions as
  for all $a\in \Sigma$ :
  \[
    \epsilon_1(\sigma_1([w],a))=\epsilon_1([wa])=\delta_1(\overline{x},wa)=\delta_1(\epsilon_1([w]),a).
  \] The proofs for the other two transitions is analogous.
\end{proof}

With this we have established a Galois connection.

\begin{cor}
  $T$ is a right adjoint of $M$.
\end{cor}

Note, that the functor $T$ can also be defined to include accepting states
(as was pointed out for deterministic automata) and is language preserving.
This follows directly by unravelling the definitions so we omit the details.

We can now define the functors
$\nuC$ and $\muPL$, making use of the reachability-inclusion adjunction, the
transition-machine adjunction we have established, and the lifted
contravariant powerset adjunction \cite{cruchten:2022:omegaAutomata}.

\begin{center}
\begin{tikzpicture}[modal]
\node[] at (0,0) (0) [label=above:{}] {$\text{Alg}(G_1)$};
\node[] at (3,0) (1) [label=above:{}] {$\text{Alg}_r(G_1)$};
\node[] at (6,0) (2) [label=above:{}] {$\mathcal{C}$\phantom{CC}};
\node[rotate=90] at (1.5,0) {$\dashv$};
\node[rotate=90] at (4.5,0) {$\dashv$};
\node at (3,2.4) {$\nuC: \text{Alg}(G_1)\to \text{Alg}(G_1)$};
\node at (3,1.8) {$\nuC= I\circ M\circ T\circ R$};
\node[] at (8,0) (5) [label=above:{}] {$\text{CoAlg}(F_2)$};
\node[] at (11,0) (6) [label=above:{}] {$\text{Alg}(G_1)^{\text{op}}$};
\node[] at (14,0) (7) [label=above:{}] {$\mathcal{C}^{\text{op}}$};
\node[rotate=90] at (9.5,0) {$\vdash$};
\node[rotate=90] at (12.5,0) {$\vdash$};
\node at (11,2.4) {$\muPL: \text{CoAlg}(F_2)\to \text{CoAlg}(F_2)$};
\node at (11,1.8) {$\muPL= \overline{P}^{\text{op}}\circ \nuC^{\text{op}}\circ \overline{P}$};
\path[->] (0) edge[bend left] node[]{$R$} (1);
\path[->] (1) edge[bend left] node[]{$T$} (2);
\path[->] (1) edge[bend left] node[]{$I$} (0);
\path[->] (2) edge[bend left] node[]{$M$} (1);
\path[->] (5) edge[bend left] node[]{$\widehat{P}$} (6);
\path[->] (6) edge[bend left] node[]{$(T\circ R)^{\text{op}}$} (7);
\path[->] (6) edge[bend left] node[]{$\widehat{P}^{\text{op}}$} (5);
\path[->] (7) edge[bend left] node[]{$(I\circ M)^{\text{op}}$} (6);
\end{tikzpicture}
\end{center}

As before, $\nuC$ is a (idempotent) comonad and $\muPL$ is a monad. Moreover,
we can extend the definitions of $\nuC$ and $\muPL$ to make them endofunctors
on the category $\text{LA}$ of lasso automata. For each initial state, there
is a canonical choice of initial state for $\muPL(X,c)$ and for each selection
of final states, there is a canonical choice of final states for
$\nuC(X,\overline{x})$.

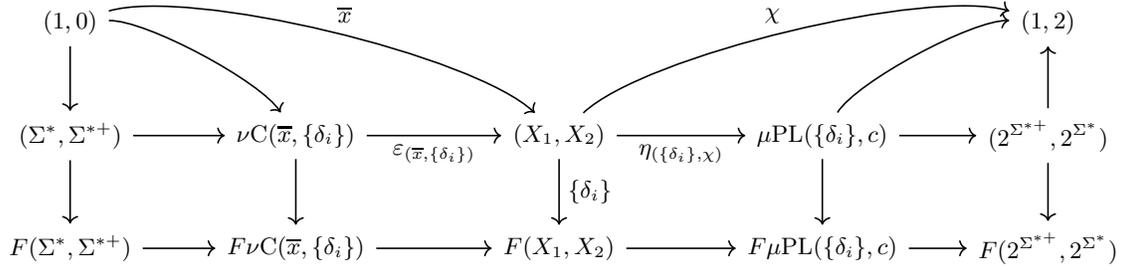
\begin{center}
\begin{tikzpicture}[modal]

\node[] at (1,0) (0) [label=above:{}] {$F(\Sigma^\ast, \Sigma^{\ast +})$};
\node[] at (7.5,0) (3) [label=above:{}] {$F(X_1,X_2)$};
\node[] at (14,0) (6) [label=above:{}] {$F(2^{\Sigma^{\ast +}}, 2^{\Sigma^\ast})$};
\node[] at (1,1.5) (7) [label=above:{}] {$(\Sigma^\ast, \Sigma^{\ast +})$};
\node[] at (7.5,1.5) (a) [label=above:{}] {$(X_1,X_2)$};
\node[] at (14,1.5) (d) [label=above:{}] {$(2^{\Sigma^{\ast +}}, 2^{\Sigma^\ast})$};
\node[] at (1,3) (e) [label=above:{}] {$(1,0)$};
\node[] at (14,3) (f) [label=above:{}] {$(1,2)$};

\node[] at (4,0) (2) [label=above:{}] {$F\nuC(\overline{x},\{\delta_i\})$};
\node[] at (11,0) (4) [label=above:{}] {$F\muPL(\{\delta_i\},c)$};
\node[] at (4,1.5) (9) [label=above:{}] {$\nuC(\overline{x},\{\delta_i\})$};
\node[] at (11,1.5) (b) [label=above:{}] {$\muPL(\{\delta_i\},c)$};
\path[->] (0) edge node[]{} (2);
\path[->] (2) edge node[]{} (3);
\path[->] (3) edge node[]{} (4);
\path[->] (4) edge node[]{} (6);
\path[->] (7) edge node[]{} (0);
\path[->] (7) edge node[]{} (9);
\path[->] (9) edge node[]{} (2);
\path[->] (9) edge node[below]{$\epsilon_{(\overline{x},\{\delta_i\})}$} (a);
\path[->] (a) edge node[]{$\{\delta_i\}$} (3);
\path[->] (a) edge node[below]{$\eta_{(\{\delta_i\},\chi)}$} (b);
\path[->] (a) edge[bend left,looseness=0.5] node[]{$\chi$} (f);
\path[->] (b) edge node[]{} (4);
\path[->] (b) edge node[]{} (d);
\path[->] (b) edge[bend left,looseness=0.5] node[]{} (f);
\path[->] (d) edge node[]{} (6);
\path[->] (d) edge node[]{} (f);
\path[->] (e) edge node[]{} (7);
\path[->] (e) edge[bend left,looseness=0.5] node[]{} (9);
\path[->] (e) edge[bend left,looseness=0.5] node[]{$\overline{x}$} (a);
\end{tikzpicture}

\begin{align*}
  \epsilon_{(\overline{x},\delta)}: \nuC (\overline{x},\delta) &\longrightarrow (\overline{x},\delta) & \eta_{(\delta,c)}: (\delta,c) &\longrightarrow \muPL (\delta,c) \\
  [u] &\longmapsto \epsilon_{1}([u]) = \delta_1(\overline{x})(u) & x &\longmapsto \eta_{1}(x) = \{[(u,av)]\mid \delta(x)(v^r,au^r)\in c\} \\
  [(u,v)] &\longmapsto \epsilon_{2}([(u,v)]) = \delta(\overline{x})(u,v) & y &\longmapsto \eta_{2}(y) = \{[u]\mid \delta_3(y)(u^r)\in c\} 
\end{align*}
\end{center}

We start by showing that $\mathsf{Eq}(\left<L \right>)$ is the syntactic congruence
of the lasso language $L$.

\begin{prop}
  Let $L$ be a lasso language. Then $\nuC(\left<L \right>)\cong \left<L \right>$.
\end{prop}

\begin{proof}
  As there is a map $\nuC(\left<L \right>)\to \left<L \right>$ and $\left<L \right>$ 
  is minimal, it is sufficient to show that $\nuC(\left<L \right>)$ is minimal.
  Let $[w]\in \nuC(\left<L \right>)$ and $(u,v)\in \Sigma^{\ast +}$. Then
  \begin{align*}
    (u,v)\in L([w]) &\iff \epsilon\in \epsilon_2(\sigma([w],(u,v))) \\
                    &\iff \epsilon \in \xi(L,(wu,v)) \\
                    &\iff \epsilon\in \{w'\mid (wu,vw')\in L\} \\
                    &\iff (wu,v)\in L.
  \end{align*}
  Hence for $[w],[w']$ we have
  \begin{align*}
    L([w])=L([w']) &\iff \forall (u,v)\in\Sigma^{\ast +}: (wu,v)\in L \iff (w'u,v)\in L \\
                   &\iff \forall (u,v)\in\Sigma^{\ast +}: (u,v)\in \xi_1(L,w) \iff (u,v)\in \xi_1(L,w') \\
                   &\iff \xi_1(L,w) = \xi_1(L,w') \\
                   &\iff (w,w')\in\ker \xi_1^\sharp \\
                   &\iff [w]=[w'].
  \end{align*}
  Hence $\nuC(\left<L \right>)\cong \left<L \right>$.
\end{proof}

For a lasso language $L$, we can now state its Myhill-Nerode equivalence
($\sim_L$) and its syntactic congruence ($\equiv_L$). The Myhill-Nerode
equivalence ${\sim_L}=(\sim_L^1,\sim_L^2)$ is obtained through the reachability (the unique map from the
initial $G_1$-algebra to $\left<L \right>$) and observability map (the unique map
from $\left<L \right>$ to the final $F_2$-coalgebra).
\begin{align*}
  w \sim^1_L w' &\iff \forall (u,v)\in\Sigma^{\ast +}: (wu,v)\in L \iff (w'u,v)\in L,\\
  (u,v)\sim^2_L (u',v') &\iff \forall w\in\Sigma^\ast: (u,vw)\in L \iff (u',v'w)\in L.
\end{align*}
The syntactic congruence corresponds to $\mathsf{Eq}(\left<L \right>)$ and
is given by $\equiv_L=(\equiv_L^1,\equiv_L^2)$ defined as
\begin{align*}
  w \equiv^1_L w' &\iff \forall u\in\Sigma,\forall (v_1,v_2)\in \Sigma^{\ast +}: (uwv_1,v_2)\in L \iff (uw'v_1,v_2)\in L,\\
  (u,v)\equiv^2_L (u',v') &\iff \forall w,w'\in \Sigma: (wu,vw')\in L \iff (wu',v'w')\in L.
\end{align*}
The Myhill-Nerode equivalence for lasso languages can already be found in
\cite{ciancia:2019:omegaAutomata,cruchten:2022:omegaAutomata}.

The relationship between $\mathsf{free}(X)=(\Sigma^\ast,\Sigma^{\ast +})/{\mathsf{Eq}(X)}$ 
and $\nuC(X,\overline{x})$ when $\left<x \right>= X$ is still the same as
it was in the case for DFAs.

\begin{prop}
  Let $(X,\overline{x})$ be a reachable pointed lasso automaton. Then
  $\nuC(X,\overline{x})\cong \mathsf{free}(X)$.
\end{prop}

\begin{proof}
  This follows simply from reachability and Lemma \ref{lem:Eq}.
\end{proof}

\begin{prop}
  For any accepting lasso automaton $(X,\{\delta_i\} ,c)$,
  $\muPL(X,\{\delta_i\} ,c)$ is minimal.
\end{prop}

\begin{proof}
  Let $P$ be a state in $\muPL(X)_1$, i.e. $P$ is a set of equivalence classes
  over $\Sigma^{\ast +}$. Then
  \[
    (u,av)\in L(P) \iff [\epsilon]\in\widehat{\sigma}(P)(u,av) \iff (v^r,au^r)\in P.
  \] 
  It follows that $P=Q\implies L(P)=L(Q)$, i.e. $\muPL(X)$ is minimal.
\end{proof}

\subsection{$\Omega$-Automata and Wilke Algebras}

For $\omega$-languages (and $\infty$-languages), the coalgebraic counterpart
to deterministic automata can be played by $\Omega$-automata, and the algebraic
counterpart to monoids can be played by Wilke algebras \cite{wilke:1993:algebraic}
(or $\omega$-semigroups).

In this section, we show that the transition Wilke algebra construction from
\cite{cruchten:2022:omegaAutomata} is functorial, and that it forms the right
adjoint of a Galois connection. The definition of a Wilke algebra and the
surrounding algebraic language theory can be found in \cite{wilke:1993:algebraic}.
The functors $F, G, F_2$ and $G_1$ in this section are the same as defined at
the start of Section \ref{section:LassoAutomata}.

Before we define the transition functor, we introduce some additional definitions.
We make use of the notion of an admissible set as
defined in \cite{cruchten:2022:omegaAutomata}. For a pointed lasso automaton
$(X_1,X_2,\overline{x},\delta_1,\delta_2,\delta_3)$, a set $c\subseteq X_2$ is
admissible if it turns the pointed lasso automaton into a pointed and
accepting $\Omega$-automaton. The set of all admissible subsets is written
$\text{Adm}(X_2)$. We define the following equivalence relation on
ultimately periodic words:
\[
  uv^\omega \sim u'v'^\omega \iff \forall x\in X_1, \forall c\in \text{Adm}(X_2): \delta(x,(u,v))\in c \iff \delta(x,(u',v'))\in c.
\] 
This is well defined as the sets are admissible, so by definition if
$(u,v)\sim_\gamma (u',v')$ then for all $x\in X_1$ we have that
$\delta(x,(u,v))\in c \iff \delta(x,(u',v'))\in c$.

\begin{rmk}
  Our definition of $\sim$ is very closely related to $\ker \delta^\sharp$ which
  we have used for lasso automata. Note that for an $\Omega$-automaton,
  we only consider certain sets as admissible as an $\Omega$-automaton
  should not be able to distinguish
  between $\gamma$-equivalent lassos. However, for lasso automata, any
  subset is admissible in this sense. If we change the notion of admissibility
  in the definition of $\sim$ to include all subsets, we recapture precisely
  the congruence $\ker\delta^\sharp$. In that sense, the transition Wilke algebra
  functor we define below is very strongly related to the transition construction
  for lasso automata.
\end{rmk}

We define a transition functor from the category of reachable $G_1$-algebras
to the category of Wilke algebra congruences $\mathcal{C}$. As before, both
categories are thin categories.

\begin{lem}
  The categories $\text{Alg}_r(G_1)$ and $\mathcal{C}$ are thin.
\end{lem}

\begin{proof}
  The proof is analogous to that of Lemma \ref{lem:thin} or \ref{lem:thin2}.
\end{proof}

The transition functor is closely related to the Wilke algebra construction
found in \cite{cruchten:2022:omegaAutomata}. The object part is the same, but
the presentation adapted and given in terms of Wilke algebra congruences.

\begin{center}
  \begin{minipage}{0.60\textwidth}
    \[
    T : \text{Alg}_r(G_1)\to \mathcal{C}
    \] 
    \begin{itemize}
      \item Objects:
        $(\overline{x},\delta_1,\delta_2,\delta_3) \mapsto (\ker \delta_1^\sharp\cap \ker \delta_\circ^\sharp\cap \ker\delta_3^\sharp,\sim)$ where
        $\sim$ is the equivalence relation from above.
      \item Morphisms:
          $h:(\overline{x},X)\twoheadrightarrow (\overline{y},Y) \implies T(\overline{x},X) \subseteq T(\overline{y},Y)$
    \end{itemize}
  \end{minipage}
\end{center}

\begin{lem}
  The functor $T$ is well-defined.
\end{lem}

\begin{proof}
  We show that $(\ker \delta_1^\sharp\cap \ker \delta_\circ^\sharp\cap \ker\delta_3^\sharp,\sim)$
  is a congruence on the free Wilke algebra over $\Sigma$, $\Sigma^{+,\ast +}$.
  For $u,v\in \Sigma^+$, we define $u\sim_\delta v :\iff (u,v)\in \ker \delta_1^\sharp\cap \ker \delta_\circ^\sharp\cap \ker\delta_3^\sharp$.
  Let $u\sim_\delta u'$ and $v\sim_\delta v'$. We have to show that
  $uv\sim_\delta u'v'$ but restrict our proof to showing that
  $(uv,u'v')\in\ker\delta_\circ$ as the other cases are similar. Let
  $x\in X_1$, then
  \[
    \delta_\circ(x,uv) = \delta_\circ(\delta_1(x,u),v) = \delta_\circ(\delta_1(x,u'),v')=\delta_\circ(x,u'v').
  \] 
  For $(-)^\omega$, let $u\sim_\delta v$ so that for all $x\in X_1$ we have
  $\delta_\circ(x,u)=\delta_\circ(x,v)$. We have to show that
  $u^\omega\sim v^\omega$. Let $x\in X_1$ and $c\in \text{Adm}(X_2)$. Then
   \[
  \delta(x,(\epsilon,u))\in c \iff \delta_\circ(x,u)\in c \iff \delta_\circ(x,v)\in c \iff \delta(x,(\epsilon,v))\in c.
  \] 
  For the mixed multiplication, let $u\sim_\delta u'$ and
  $vw^\omega\sim v'w'^\omega$. For $x\in X_1$ and $c\in \text{Adm}(X_2)$ we have
  that
   \[
  \delta(x,(uv,w))\in c \iff \delta(\delta_1(x,u),(v,w))\in c \iff \delta(\delta_1(x,u'),(v',w'))\in c \iff \delta(x,(u'v',w'))\in c.
  \] 
  Finally, the pumping and rotation law hold trivially as we are working
  with admissible sets. 

  Next let $h:(\overline{x},X)\twoheadrightarrow (\overline{y},Y)$. We claim
  that $T(\overline{x},X) \subseteq T(\overline{y},Y)$. This is easy to see
  for  $\ker\delta_1$, $\ker\delta_\circ$ and $\ker\delta_3$, so we only show
  it for $\sim$. In order to do so, we make a small remark about admissible
  sets. If $c\in \text{Adm}(Y_2)$, then $h^\ast c=\{x\in X_2\mid h_2(x)\in c\}\in \text{Adm}(X_2)$.
  This is shown using reachability and is straight-forward. With this, let
  $uv^\omega\sim_X u'v'^\omega$,  $y\in Y_1$, $x\in X_1$ with $h_1(x)=y$ and
  $c\in \text{Adm}(Y_2)$. Then
  \[
    \delta(y,(u,v))\in c \iff \delta(x,(u,v))\in h^\ast c \iff \delta(x,(u',v'))\in h^\ast c \iff \delta(y,(u',v'))\in c.
  \] Hence $uv^\omega \sim_Y u'v'^\omega$.
\end{proof}

In order to show that $T$ is a right adjoint, we do not explicitly construct
the machine functor, but instead apply the adjoint functor theorem for
preorders, for which we only have to show that $\text{Alg}_r(G_1)$ has
arbitrary meets, and that $T$ preserves them.

\begin{prop}
  The category $\text{Alg}_r(G_1)$ is complete.
\end{prop}

\begin{proof}
  As the category $\text{Alg}_r(G_1)$ is a preorder, we show that it has
  arbitrary meets. The empty meet is just the singleton set $1$ with trivial
  transitions and the only possible initial state. Given a set of reachable $G_1$ algebras
  $\{(X_{i,1},X_{i,2},\overline{x_i},\delta_1^i,\delta_2^i,\delta_3^i)\}_{i\in I}$,
  we get its meet by taking the product of the state spaces, defining transitions
  pointwise and the initial state as the tuple consisting of initial states $\overline{x_i}$, and
  then taking the reachable part.
\end{proof}

\begin{prop}
  The functor $T$ preserves arbitrary meets.
\end{prop}

\begin{proof}
  Let $\{(X_{i,1},X_{i,2},\overline{x_i},\delta_1^i,\delta_2^i,\delta_3^i)\}_{i\in I}$
  be a set of reachable $G_1$ algebras with meet $X=\Pi \{X_i\}_{i\in I}$.
  
  Firstly, the meets in $\mathcal{C}$ are given by intersection of congruences,
  which are well-defined. We now have to show that
  \[
    T(X)=\bigcap_{i\in I}T(X_i).
  \] 
  This is straight-forward for the first sort, i.e. we have
  \[
    \ker\delta_{X,1}\cap \ker\delta_{X,\circ}\cap\ker\delta_{X,3} = \bigcap_{i\in I}\ker\delta_{X_i,1}\cap \ker\delta_{X_i,\circ}\cap\ker\delta_{X_i,3}.
  \] 
  For $u,v\in \Sigma^+$ we have
  \begin{align*}
    u\sim_{\delta,X}v &\iff \forall \vec{x}, \forall {\Box} \in \{1,\circ,3\} : \delta_{X,\Box}(x,u)=\delta_{X,\Box}(x,v) \\
                      &\iff \forall i\in I, \forall x_i, \forall {\Box} \in \{1,\circ,3\} : \delta_{X_i,\Box}(x_i,u)=\delta_{X_i,\Box}(x_i,v) \\
                      &\iff \forall i\in I: u\sim_{\delta,X_i} v.
  \end{align*}
  For the second sort, we need the following observation. For all $c\in \text{Adm}(X_2)$
  and $(u,v)\sim_\gamma (u',v')$:
  \begin{align*}
    \forall i\in I: \delta_i(x_i,(u,v)) \in \pi_i(c) &\iff \delta(\vec{x},(u,v))\in c \\
                                                     &\iff \delta(\vec{x},(u',v'))\in c \\
                                                     &\iff \forall i\in I: \delta_i(x_i,(u',v')) \in \pi_i(c).
  \end{align*} Moreover, for any $c_i\in \text{Adm}(X_{i,2})$ there exists some $c\in \text{Adm}(X_2)$ 
  such that $c_i=\pi_i(c)$ (take for instance the cartesian product of $c_i$ with
  $X_{2,j}$ where $j\not =i$ and intersect with $X_2$).
  For $uv^\omega, u'v'^\omega\in \Sigma^{\text{up}}$, we then have that
  \allowdisplaybreaks{
  \begin{align*}
    uv^\omega\sim_X u'v'^\omega &\iff \forall \vec{x}\in X_1, \forall c\in \text{Adm}(X_2): \delta_X(\vec{x} ,(u,v))\in c \iff \delta_X(\vec{x} ,(u',v'))\in c \\
                                &\iff \forall w\in \Sigma^\ast, \forall c\in \text{Adm}(X_2): \delta_X(\delta_1(\vec{\overline{x}_i} ,w),(u,v))\in c \iff \delta_X(\delta_1(\vec{\overline{x}_i},w) ,(u',v'))\in c \\
                                &\iff \forall i\in I, \forall w\in \Sigma^\ast, \forall c\in \text{Adm}(X_2): \\
                                &\qquad\qquad \delta_{X_i}(\delta_{X_i,1}(\overline{x}_i ,w),(u,v))\in \pi_i(c) \iff \delta_{X_i}(\delta_{X_i,1}(\overline{x}_i ,w),(u',v'))\in \pi_i(c)  \\
                                &\iff \forall i\in I, \forall w\in \Sigma^\ast, \forall c_i\in \text{Adm}(X_{i,2}): \\
                                &\qquad\qquad \delta_{X_i}(\delta_{X_i,1}(\overline{x}_i ,w),(u,v))\in c_i \iff \delta_{X_i}(\delta_{X_i,1}(\overline{x}_i ,w),(u',v'))\in c_i  \\
                                &\iff \forall i\in I, uv^\omega \sim_{X_i} u'v'^\omega.\qedhere
  \end{align*}
}
\end{proof}     

From the adjoint functor theorem for preorders it follows that $T$ is a
right adjoint.

\begin{cor}
  The functor $T$ is a right adjoint.
\end{cor}

\section{Conclusion}

In this paper, we investigated the well-known transition-machine construction
and showed that in the classical setting, this construction gives rise to
an adjunction. Based on this adjunction, we drew a close link to sets of
equations and coequations. We obtained a comonad which maps an
automaton to the greatest set of equations it satisfies on its reachable part.
We also obtained a monad which maps an automaton to the least preformation of
languages which includes certain languages one can obtain from varying the
initial and final states. These deserve further investigation.

Furthermore, we showed that transition constructions, which form Galois connections,
can also be constructed for lasso and $\Omega$-automata. For lasso automata
in particular, we defined sets of equations and coequations and made links to
the Myhill-Nerode and syntactic congruence of a lasso language.

Our work presents directions for future work such as the
exploration of the link to work on minimisation
\cite{bonchi:2014:algebraCoalgebra,bezhanishvili:2020:minimisation} which
is used in the construction of $\muPL$.

\paragraph{Acknowledgements.} The author would like to thank Harsh Beohar and
Georg Struth for valuable discussions, and also Anton Chernev for the various
discussions specifically on the adjunctions surrounding the transition
constructions for lasso and $\Omega$-automata.

\bibliographystyle{plain} 
\bibliography{references}

\end{document}